\documentclass[conference,a4paper]{IEEEtran}
\usepackage{enumerate}
\usepackage{epstopdf}
\usepackage{cite}
\usepackage{amscd}
\usepackage{dsfont}
\usepackage{latexsym}
\usepackage{array}

\usepackage{tikz}\usetikzlibrary{calc}
\usepackage{tabularx}
\usepackage{epsfig}
\usepackage{graphicx,subfigure}
\usepackage{epstopdf}
\usepackage{graphicx}
\usepackage{amsmath,mathrsfs}
\usepackage{amsthm}
\usepackage{amssymb}  
\usepackage{amsbsy}
\usepackage{color}
\usepackage{stfloats}
\usepackage{algorithm}
\usepackage{algpseudocode}
\usepackage{bm}
\usepackage{flushend}

\newtheoremstyle{note}
{3pt}
{1pt}
{}
{\parindent}
{\itshape}
{:}
{.5em}
{\thmname{#1}\thmnumber{ #2}\thmnote{\thmnote{ (#3)}}}
\theoremstyle{note}

\newtheorem{theorem}{Theorem}
\newtheorem{lemma}{Lemma}
\newtheorem{remark}{Remark}
\newtheorem{definition}{Definition}

\theoremstyle{definition}

\newtheoremstyle{dotless}{}{}{\itshape}{}{\bfseries}{}{ }{}
\theoremstyle{dotless}

\newcommand {\aplt} {\ {\raise-.5ex\hbox{$\buildrel<\over{\mbox{\scriptsize $\sim$}}$}}\ }

\begin{document}
%
\title{Polar Lattices for Strong Secrecy Over the
Mod-$\Lambda$ Gaussian Wiretap Channel}

\author{\IEEEauthorblockN{Yanfei Yan, Ling Liu and Cong Ling}
\IEEEauthorblockA{Department of Electrical and Electronic Engineering\\
Imperial College London\\
London, UK\\
Email: \{y.yan10, l.liu12\}@imperial.ac.uk, cling@ieee.org}
}
\maketitle

\begin{abstract}
Polar lattices, which are constructed from polar codes, are provably good for the additive white Gaussian noise (AWGN) channel. In this work,
we propose a new polar lattice construction that achieves the secrecy capacity under the strong secrecy criterion over the mod-$\Lambda$ Gaussian wiretap channel. This construction leads to an AWGN-good lattice and a secrecy-good lattice simultaneously. The design methodology is mainly based on the equivalence in terms of polarization between the $\Lambda/\Lambda'$ channel in lattice coding and the equivalent channel derived from the chain rule of mutual information in multilevel coding.
\end{abstract}


%
\IEEEpeerreviewmaketitle

\section{Introduction}
Wyner \cite{wyner1} introduced the wiretap channel and showed that both
reliability to transmission errors and a prescribed degree of data
confidentiality could be attained by channel coding without any key
bits if the channel between the
sender and the eavesdropper (wiretap channel $C_{W}$) is a degraded version of the channel
between the sender and the legitimate receiver (main channel $C_{V}$). The goal is to design a coding scheme that makes it possible to communicate both reliably and securely, as the block length of transmitted codeword $N$ tends to infinity. Reliability is measured by the decoding error probability of the legitimate user, namely $\lim_{N\rightarrow\infty} \text{Pr}\{\hat{M}\neq M\}=0$, where $M$ is the confidential message and $\hat{M}$ is its estimation. Secrecy is measured by the mutual information between $M$ and the signal received by the eavesdropper $Z^N$. Currently the widely accepted strong secrecy condition was proposed by Csisz\'{a}r \cite{csis1}: $\lim_{N\rightarrow\infty}I(M;Z^{N})=0$. In simple terms, the secrecy capacity is the maximum achievable rate of any coding scheme that can satisfy both the reliability and strong secrecy conditions.

Polar codes \cite{polarcodes} have been shown their great potential of solving this wiretap coding problem. The polar coding scheme proposed in \cite{polarsecrecy} is proved to achieve the strong secrecy capacity with explicit construction when $C_V$ and $C_W$ are both binary-input memoryless channels, although it is not able to guarantee the reliability condition. A subsequently modified scheme \cite{NewPolarSchemeWiretap} fixes this issue and finally satisfies the two conditions. However, for continuous channels such as the Gaussian wiretap channel, the problem of achieving strong secrecy with a practical code is still open.

There has been some progress in wiretap lattice coding for the Gaussian wiretap channel. On the theoretical aspect, the achievable rate for lattice coding achieving weak secrecy over the Gaussian wiretap channel has been derived \cite{cong}. Furthermore, the existence of lattice codes approaching the secrecy capacity under the strong secrecy criterion (semantic security) was demonstrated
in \cite{cong2}. On the practical aspect, wiretap lattice codes were proposed in \cite{belf3} to maximize the eavesdropper's decoding error probability. Since the analysis of the mod-$\Lambda$ Gaussian channel is a key element to the analysis of Gaussian channels and it provides considerable insight into the construction for the Gaussian wiretap channel, we limit ourselves to the mod-$\Lambda$ Gaussian wiretap channel shown in Fig. \ref{fig:wiretapchannel}. The lattice coset coding scheme for the mod-$\Lambda_s$ Gaussian wiretap channel was introduced in \cite{cong2}. $M$ is encoded into an $N$-dimensional transmitted codeword $X^N$.
The outputs $Y^{N}$ and $Z^{N}$ at Bob and Eve's end
respectively are given by
\begin{eqnarray}
\left\{\begin{aligned}
&Y^{N}=[X^{N}+W_{b}^{N}]\text{ mod } \Lambda_s, \notag\\
&Z^{N}=[X^{N}+W_{e}^{N}] \text{ mod } \Lambda_s, \notag\
\end{aligned}\right.
\notag\
\end{eqnarray}
where $W_{b}^{N}$ and $W_{e}^{N}$ are $N$-dimensional Gaussian vectors with zero mean and variance
$\sigma_{b}^{2}$, $\sigma_{e}^{2}$ respectively. The channel input $X^N$ satisfies the power constraint
\begin{eqnarray}
\frac{1}{N}E\Big[\parallel X^N\parallel^2\Big]\leq P.
\notag\
\end{eqnarray}
To satisfy this power constraint, we choose a shaping lattice $\Lambda_s$ whose second moment per dimension is $P$. Under the continuous approximation for large constellations, the transmission power will be equal $P$.

\begin{figure}[t]
    \centering
    \includegraphics[width=8cm]{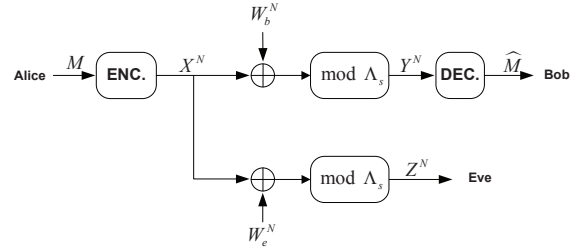}
    \caption{The mod-$\Lambda_s$ Gaussian wiretap channel.}
    \label{fig:wiretapchannel}
\end{figure}

Assume that $\Lambda_s$, $\Lambda_e$ and $\Lambda_b$ are quantization-good \cite{BK:Zamir}, secrecy-good and AWGN-good respectively (the latter two will be defined in Sect. II). Let $\Lambda_s\subset\Lambda_e\subset\Lambda_b$ be a nested chain of $N$-dimensional lattices in $\mathbb{R}^N$ such that $\frac{1}{N}\log|\Lambda_b/\Lambda_e|=R$. Consider a one-to-one mapping: $M\rightarrow[\Lambda_b/\Lambda_e]$ which associates each message to a coset leader $\lambda_m\in\Lambda_b/\Lambda_e$. Alice selects a random lattice point $\lambda\in\Lambda_e\cap \mathcal{V}(\Lambda_s)$ ($\mathcal{V}(\Lambda)$ is the Voronoi region of $\Lambda$ defined in Sect. II) and transmits $X^N=\lambda+\lambda_m$. This coding scheme can achieve both reliability and strong secrecy. Although the existence of secrecy-good lattices has been proved, their explicit construction is still missing.

Polar lattices, which can be considered as the counterpart of polar codes in the Euclidean space, is a promising candidate for the Gaussian wiretap channel. They have already been proved to be AWGN-good \cite{yan2}, which means the reliability condition can be satisfied and we just need to take the strong secrecy into account. Motivated by \cite{polarsecrecy}, we propose an explicit polar lattice construction which can be proved to achieve both the strong secrecy and reliability over the mod-$\Lambda_s$ channel. Conceptually this new polar lattice can be regarded as an AWGN-good lattice $\Lambda_b$ nested within a secrecy-good lattice $\Lambda_e$. The design philosophy is mainly based on the equivalence between two well-known channels which greatly simplifies the construction for secrecy-good lattices. We note that the mod-$\Lambda$ front-end simplifies the problem under study in this paper. The coding system with a shaping method to achieve the secrecy capacity of the genuine Gaussian wiretap channel will be addressed in the journal paper.

The paper is organized as follows:
Section II presents the background of lattices. The relationship between two types of channels is investigated in Section III.
In Section IV, by using the equivalence proved in the previous section, we propose
our new polar lattice coding scheme and prove its secrecy and reliability. Some remarks are given in Section V.

\section{Background on Lattices}
\subsection{Definitions}
A lattice is a discrete subgroup of $\mathbb{R}^{n}$ which can be described by
\begin{eqnarray}
\Lambda=\{\bm \lambda=\mathbf{B}\mathbf{x}:\mathbf{x}\in\mathbb{Z}^{n}\}, \notag\
\end{eqnarray}
where the columns of the generator matrix $\mathbf{B}=[\mathbf{b}_{1}, \cdots, \mathbf{b}_{n}]$ are linearly independent.

For a vector $\mathbf{x}\in\mathbb{R}^{n}$, the nearest-neighbor quantizer associated with $\Lambda$ is $Q_{\Lambda}(\mathbf{x})=\text{arg}\min_{\bm \lambda\in\Lambda}\parallel\bm\lambda-\mathbf{x}\parallel$. We define the modulo lattice operation by $\mathbf{x} \text{ mod }\Lambda\triangleq \mathbf{x}-Q_{\Lambda}(\mathbf{x})$. The Voronoi region of $\Lambda$, defined by $\mathcal{V}(\Lambda)=\{\mathbf{x}:Q_{\Lambda}(\mathbf{x})=0\}$, specifies the nearest-neighbor decoding region. The Voronoi cell is one example of fundamental region of the lattice. A measurable set $\mathcal{R}(\Lambda)\subset\mathbb{R}^{n}$ is a fundamental region of the lattice $\Lambda$ if $\cup_{\bm\lambda\in\Lambda}(\mathcal{R}(\Lambda)+\bm\lambda)=\mathbb{R}^{n}$ and if $(\mathcal{R}(\Lambda)+\bm\lambda)\cap(\mathcal{R}(\Lambda)+\bm\lambda')$ has measure 0 for any $\bm\lambda\neq\bm\lambda'$ in $\Lambda$. The volume of a fundamental region is equal to that of the Voronoi region $\mathcal{V}(\Lambda)$, which is given by $V(\Lambda)=\mid\text{det}(\mathbf{B})\mid$.

For $\sigma>0$, we define the noise distribution of the AWGN channel with zero mean and variance $\sigma^{2}$ as
\begin{eqnarray}
f_{\sigma}(\mathbf{n})=\frac{1}{(\sqrt{2\pi}\sigma)^{n}}e^{-\frac{\parallel \mathbf{n}\parallel^{2}}{2\sigma^{2}}}, \notag\
\end{eqnarray}
for all $\mathbf{n}\in\mathbb{R}^{n}$. Given $\sigma$, the volume-to-noise ratio (VNR) of an $n$-dimension lattice $\Lambda$ is defined by
\begin{eqnarray}
\gamma_{\Lambda}(\sigma)\triangleq\frac{V(\Lambda)^\frac{2}{n}}{\sigma^2}. \notag\
\end{eqnarray}

We also need the $\Lambda$-periodic function
\begin{eqnarray}
f_{\sigma,\Lambda}(\mathbf{n})=\sum\limits_{\bm\lambda\in\Lambda}f_{\sigma,\bm\lambda}(\mathbf{n})=\frac{1}{(\sqrt{2\pi}\sigma)^{n}}\sum\limits_{\bm\lambda\in\Lambda}e^{-\frac{\parallel \mathbf{n}-\bm\lambda\parallel^{2}}{2\sigma^{2}}}, \notag\
\end{eqnarray}
for all $\mathbf{n}\in\mathbb{R}^n$.

We note that $f_{\sigma,\Lambda}(\mathbf{n})$ is a probability density function (PDF) if $\mathbf{n}$ is restricted to the the fundamental region $\mathcal{R}(\Lambda)$. This distribution for $\mathbf{n}\in\mathcal{R}(\Lambda)$ is actually the PDF of the $\Lambda$-aliased Gaussian noise, i.e., the Gaussian noise after the mod-$\Lambda$ operation \cite{forney6}. When $\sigma$ is small, the effect of aliasing becomes insignificant and the $\Lambda$-aliased Gaussian density $f_{\sigma,\Lambda}(\mathbf{n})$ approaches the Gaussian distribution. When $\sigma$ is large, $f_{\sigma,\Lambda}(\mathbf{n})$ approaches the uniform distribution.

In this paper, we mainly deal with the following two kinds of lattices.
\begin{definition}[AWGN-good]
A lattice $\Lambda_b$ is called AWGN-good if it is capable of achieving the Poltyrev capacity \cite{poltyrev} with a fast-vanishing error probability as long as its VNR is larger than $2\pi e$.
\end{definition}

\begin{definition}[Secrecy-good]
A lattice $\Lambda_e$ which results in fast-vanishing information leakage $I(M;Z^{N})$ is regarded as a secrecy-good lattice.
\end{definition}
Note that this definition is different from that in \cite{cong2}, which is based on the flatness factor.

\subsection{Mod-$\Lambda$ and $\Lambda/\Lambda'$ Channel}
A sublattice $\Lambda' \subset \Lambda$ induces a partition (denoted by $\Lambda/\Lambda'$) of $\Lambda$ into equivalence groups modulo $\Lambda'$. The order of the partition is denoted by $|\Lambda/\Lambda'|$, which is equal to the number of the cosets. If $|\Lambda/\Lambda'|=2$, we call this a binary partition. Let $\Lambda_{1}/\cdots/\Lambda_{r-1}/\Lambda_{r}$ for $r \geq 2$ be an $n$-dimensional lattice partition chain. If only one level is applied ($r=2$), the construction is known as ``Construction A". If multiple levels are used, the construction is known as ``Construction D" \cite[p.232]{yellowbook}. For each
partition $\Lambda_{\ell}/\Lambda_{\ell+1}$ ($1\leq \ell \leq r-1$) a code $C_{\ell}$ over $\Lambda_{\ell}/\Lambda_{\ell+1}$
selects a sequence of coset representatives $a_{\ell}$ in a set $A_{\ell}$ of representatives for the cosets of $\Lambda_{\ell+1}$.

A lattice constructed by the binary partition chain is referred as the binary lattice. This construction requires a set of nested linear binary codes.
The set of nested linear binary codes $C_{\ell}$ with block length $N$ and dimension of information bits $k_{\ell}$ are represented as
$[N,k_{\ell}]$ for $1\leq\ell\leq r-1$ and $C_{1}\subseteq C_{2}\cdot\cdot\cdot\subseteq C_{r-1}$. Let $\psi$ be the natural embedding of $\mathbb{F}_{2}^{N}$ into $\mathbb{Z}^{N}$, where $\mathbb{F}_{2}$ is the binary field. Let $\mathbf{b}_{1}, \mathbf{b}_{2},\cdots, \mathbf{b}_{N}$ be a basis of $\mathbb{F}_{2}^{N}$ such that $\mathbf{b}_{1},\cdots \mathbf{b}_{k_{\ell}}$ span $C_{\ell}$. When $n=1$, the binary lattice $L$ consists of all vectors of the form
\begin{eqnarray}
\sum_{\ell=1}^{r-1}2^{\ell-1}\sum_{j=1}^{k_{\ell}}\alpha_{j}^{(\ell)}\psi(\mathbf{b}_{j})+2^{r-1}\mathbf{l},
\label{constructionD}
\end{eqnarray}
where $\alpha_{j}^{(\ell)}\in\{0,1\}$ and $\mathbf{l}\in\mathbb{Z}^{N}$.

A mod-$\Lambda$ channel is a Gaussian channel with a modulo-$\Lambda$ operator in the front end \cite{multilevel1,forney6}. The capacity of the mod-$\Lambda$ channel is \cite{forney6}
\begin{eqnarray}\label{mod-capacity}
C(\Lambda, \sigma^{2})=\log  V(\Lambda)-h(\Lambda, \sigma^{2}),
\label{eqn:modchannelcapacity}
\end{eqnarray}
where $h(\Lambda, \sigma^{2})$ is the differential entropy of the $\Lambda$-aliased noise over $\mathcal{V}(\Lambda)$:
\begin{eqnarray}
h(\Lambda,\sigma^{2})=-\int_{\mathcal{V}(\Lambda)}f_{\sigma,\Lambda}(\mathbf{n})\text{ log } f_{\sigma,\Lambda}(\mathbf{n})d\mathbf{n}. \notag\
\end{eqnarray}
The differential entropy is maximized to $\log V(\Lambda)$ by the uniform distribution over $\mathcal{V}(\Lambda)$. It is known that the $\Lambda/\Lambda'$ channel (i.e., the mod-$\Lambda'$ channel whose input is drawn from $\Lambda\cap\mathcal{V}(\Lambda')$) is regular, and the optimum input distribution is uniform \cite{forney6}. Furthermore, the $\Lambda/\Lambda'$ channel is a binary-inputs memoryless symmetric channel (BMS) if $|\Lambda/\Lambda'|=2$ \cite{yan2}. The capacity of the $\Lambda/\Lambda'$ channel for Gaussian noise of variance $\sigma^2$ is given by \cite{forney6}
\begin{equation}\notag
\begin{split}
  C(\Lambda/\Lambda', \sigma^2) &= C(\Lambda', \sigma^2) - C(\Lambda, \sigma^2) \\
  &= h(\Lambda, \sigma^2) - h(\Lambda', \sigma^2) + \log (V(\Lambda')/V(\Lambda)).
\end{split}
\end{equation}

\section{Equivalent Channels}
We use a binary lattice partition chain $\Lambda_{1}/\Lambda_{2}/\cdot\cdot\cdot\Lambda_{r}$ to construct lattices so that we can use binary codes at each level. Here $\Lambda_s$ is for shaping and we assume $\Lambda_s \subset \Lambda_r^N$. Eve can do the mod-$\Lambda_r$ operation to remove all the random bits $\lambda$. With some overloading of notation, we also write $Z^N=(X^{N}+W_{e}^{N})\text{ mod }\Lambda_r^N$. It is known that $Z^N$ is a sufficient statistic for $(X^{N}+W_{e}^{N})\text{ mod }\Lambda_s$ \cite{forney6}. In other words, the mod-$\Lambda_r$ operation is information-lossless in the sense that $I(X^N;Z^N)=I(X^N;(X^{N}+W_{e}^{N})\text{ mod }\Lambda_s)$. As far as mutual information is concerned, we can use the mod-$\Lambda_r$ operator instead of the mod-$\Lambda_s$ operator here. We also assume the inputs to the binary encodes $C_1,\cdot\cdot\cdot,C_{r-1}$ are uniformly distributed and independently between each other. Therefore the signal points $X^N=(X_1^N,\cdot\cdot\cdot,X_{r-1}^N)\in \Lambda_1^N/\Lambda_r^N$ have equal a priori probabilities where $X_i^N$ is the output of the $i$-th binary encoder. And we note that all the lattice partitions are regular. The $\Lambda_i/\Lambda_{i+1}$ channel corresponding to the variance $\sigma^2$ is denoted by $W(\Lambda_i/\Lambda_{i+1},\sigma^2)$. The conditional PDF of this channel is $f_{\sigma,\Lambda_{i+1}}(z-x_i)$ \cite{forney6}.

By the chain rule of mutual information,
\begin{eqnarray}\label{eqn:chain_rule}
\begin{aligned}
I(Z;X)&=I(Z;X_1,X_2,\cdot\cdot\cdot,X_{r-1}) \\
&=I(Z;X_1)+I(Z;X_2|X_1)+\cdot\cdot\cdot \\
&+I(Z;X_{r-1}|X_1,\cdot\cdot\cdot,X_{r-2}).
\end{aligned}
\end{eqnarray}

Transmitting vectors $x$ with binary digits $x_i$, $i=1,\cdot\cdot\cdot r-1$ over the mod-$\Lambda_{r}$ channel can be separated into the parallel transmission of individual digits $x_i$ over $r-1$ equivalent channels, provided that $x_1,\cdot\cdot\cdot,x_{i-1}$ are known \cite{multilevel}. We investigate these equivalent channels, denoted by $W'(Z;X_i|X_1,\cdot\cdot\cdot,X_{i-1})$, in the framework of lattice codes.

From \cite{forney6}, the conditional PDF of $z$ is $f_Z(z|x)=f_{\sigma,\Lambda_r}(z-x)$, $z\in\mathcal{V}(\Lambda_r)$. Then based on the chain rule \eqref{eqn:chain_rule} and equivalent channels \cite[(5)]{multilevel}, the conditional PDF of the first equivalent channel with the input $x_1$ is
{\allowdisplaybreaks\begin{eqnarray}\notag
\begin{aligned}
f_{Z}(z|x_1)&=\frac{1}{\text{Pr}(\Lambda_2/\Lambda_r+x_1)}\sum_{x\in\Lambda_2/\Lambda_r+x_1}\text{Pr}(x)f_{Z}(z|x) \\
&\stackrel{(1)}=\frac{1}{|\Lambda_2/\Lambda_r+x_1|}\sum_{x\in\Lambda_2/\Lambda_r+x_1}f_{\sigma,\Lambda_r}(z-x) \\
&\stackrel{(2)}=\frac{1}{|\Lambda_2/\Lambda_r+x_1|}\sum_{x\in\Lambda_2/\Lambda_r+x_1}f_{\sigma,\Lambda_2}(z-x)\\
&\stackrel{(3)}=\frac{1}{|\Lambda_2/\Lambda_r|}\sum_{x\in\Lambda_2/\Lambda_r+x_1}f_{\sigma,\Lambda_2}(z-x),\:\: z\in\mathcal{V}(\Lambda_r),
\end{aligned}
\end{eqnarray}}
where $(1)$ is due to the uniform assumption, $(2)$ is because $\Lambda_2/\Lambda_r+x_1$ represents the mod-$\Lambda_2$ operation plus a shift $x_1$ inside $\mathcal{V}(\Lambda_r)$, $(3)$ is due to the fact that $|\Lambda_2/\Lambda_r+0|=|\Lambda_2/\Lambda_r+1|$ and $f_{\sigma,\Lambda_2}(z)$ $z\in\mathcal{V}(\Lambda_2)$ is the conditional PDF of the mod-$\Lambda_2$ channel. We note that $f_{Z}(z|x_1)$ is $\Lambda_2$-periodic. We note that this is the same result as \cite[Lemma 6]{forney6}.
\begin{remark}
If the inputs are uniformly distributed, $W'(Z;X_i|X_1,\cdot\cdot\cdot,X_{i-1})$ can also be proved to be a BMS channel in the same way that the  channel $W(\Lambda_i/\Lambda_{i+1},\sigma^2)$ is proved to be a BMS channel in \cite{yan2}.
\end{remark}

The conditional differential entropy $h(Z|X_1)$ is \cite{forney6}
\begin{eqnarray}\notag
\begin{aligned}
h(Z|X_1)&=-\int_{\mathcal{V}(\Lambda_r)}f_{Z}(z|x_1)\cdot\log f_{Z}(z|x_1)dz \\
&=\log|\Lambda_2/\Lambda_r|+h(\Lambda_2,\sigma^2).
\end{aligned}
\end{eqnarray}

By \cite[Lemma 6]{forney6}, with the uniform inputs, the differential entropy $h(Z)=\log|\Lambda_1/\Lambda_r|+h(\Lambda_1,\sigma^2)$. Therefore the mutual information of the first equivalent channel is
\begin{eqnarray}\notag
\begin{aligned}
I(Z;X_1)&=h(Z)-h(Z|X_1) \\
&=h(\Lambda_1,\sigma^2)-h(\Lambda_2,\sigma^2)+\log|\Lambda_1/\Lambda_2|,
\end{aligned}
\end{eqnarray}
which is equal to the mutual information of the $\Lambda_1/\Lambda_2$ channel $I(W(\Lambda_1/\Lambda_2,\sigma^2))$.

One can verify that the mutual information of the $i$-th equivalent channel $I(Z;X_i|X_1,\cdot\cdot\cdot,X_{i-1})$ is equal to the mutual information of the $\Lambda_i/\Lambda_{i+1}$ channel in the same fashion.

This result motivates a stronger statement, namely the polar codes constructed from that these two channels are the same. To see this, it suffices to show that the mutual information and Bhattacharyya parameters of the resultant bit-channels which are polarized from $W(\Lambda_1/\Lambda_2,\sigma^2)$ and $W'(Z;X_1)$ are the same. Let $Q(z|x)$ be a BMS channel with binary input alphabet $\mathcal{X}\in\{0,1\}$ and output alphabet $\mathcal{Z}\in\mathbb{R}$. Consider a random vector $U^2$ that is uniformly distributed over $\mathcal{X}^2$. Let $X^2=U^2\cdot\left[\begin{smallmatrix}1&0\\1&1\end{smallmatrix}\right]$ be the input to two independent copies of the channel $Q$ and let $Z^2$ be the corresponding outputs. After the channel combining and splitting, the resultant bit-channels \cite{polarcodes} are defined as
\begin{eqnarray}
\begin{aligned}
&Q_2^{(1)}(z^2|u_1)=\frac{1}{2}\sum_{u_2}Q(z_1|u_1\oplus u_2)Q(z_2|u_2), \\ \notag\
&Q_2^{(2)}(z^2,u_2|u_1)=\frac{1}{2}Q(z_1|u_1\oplus u_2)Q(z_2|u_2). \notag\
\end{aligned} \notag\
\end{eqnarray}

Then we apply this polarization transformation to $W(\Lambda_1/\Lambda_2,\sigma^2)$ and $W'(Z;X_1)$, respectively. After some mathematic manipulations, we get
\begin{eqnarray}
\begin{aligned}
W_2^{(1)}(z^2|0)&=\frac{1}{2}(f_{\sigma,\Lambda_2}(z_1)f_{\sigma,\Lambda_2}(z_2) \\
&+f_{\sigma,\Lambda_2}(z_1-1)f_{\sigma,\Lambda_2}(z_2-1)), \\ \notag\
W_2^{(1)}(z^2|1)&=\frac{1}{2}(f_{\sigma,\Lambda_2}(z_1-1)f_{\sigma,\Lambda_2}(z_2) \\
&+f_{\sigma,\Lambda_2}(z_1)f_{\sigma,\Lambda_2}(z_2-1)), \\ \notag\
\end{aligned} \notag\
\end{eqnarray}
and
\begin{eqnarray}
\begin{aligned}
W_2^{\prime(1)}&(z^2|0)= \\
&\frac{1}{2|\Lambda_2/\Lambda_r|^2}\Big(\sum_{x\in\Lambda_2/\Lambda_r}f_{\sigma,\Lambda_2}(z_1-x)f_{\sigma,\Lambda_2}(z_2-x) \\ \notag\
&+\sum_{x\in\Lambda_2/\Lambda_r}f_{\sigma,\Lambda_2}(z_1-x-1)f_{\sigma,\Lambda_2}(z_2-x-1)\Big),\\ \notag\
W_2^{\prime(1)}&(z^2|1)= \\
&\frac{1}{2|\Lambda_2/\Lambda_r|^2}\Big(\sum_{x\in\Lambda_2/\Lambda_r}f_{\sigma,\Lambda_2}(z_1-x-1)f_{\sigma,\Lambda_2}(z_2-x) \\ \notag\
&+\sum_{x\in\Lambda_2/\Lambda_r}f_{\sigma,\Lambda_2}(z_1-x)f_{\sigma,\Lambda_2}(z_2-x-1)\Big).\\ \notag\
\end{aligned} \notag\
\end{eqnarray}

By the definitions of the mutual information and the Bhattacharyya parameter of a BMS channel \cite{polarcodes}
\begin{eqnarray}
\begin{cases}
&I(Q)\triangleq\int\sum_{x}\frac{1}{2}Q(y|x)\log\frac{Q(y|x)}{\frac{1}{2}Q(y|0)+\frac{1}{2}Q(y|1)}dy \\ \notag\
&Z(Q)\triangleq\int \sqrt{Q(y|0)Q(y|1)}dy \notag\
\end{cases},
\end{eqnarray}
we have
\begin{eqnarray}
\begin{cases}
&I(W_2^{(1)}(z^2|x_1))=I(W_2^{\prime(1)}((z^2|x_1))) \\ \notag\
&Z(W_2^{(1)}(z^2|x_1))=Z(W_2^{\prime(1)}(z^2|x_1))) \notag\
\end{cases}.
\end{eqnarray}
And it is not difficult to verify that
\begin{eqnarray}
\begin{cases}
&I(W_2^{(1)}(z^2,x_2|x_1))=I(W_2^{\prime(1)}((z^2,x_2|x_1))) \\ \notag\
&Z(W_2^{(1)}(z^2,x_2|x_1))=Z(W_2^{\prime(1)}(z^2,x_2|x_1))) \notag\
\end{cases}.
\end{eqnarray}

Since the construction of polar codes are based on either the mutual information or the Bhattacharyya parameter of the bit-channels, polar codes constructed for $W(\Lambda_1/\Lambda_2,\sigma^2)$ and $W'(Z;X_1)$ are the same. The validation of the equivalence between the $i$-th channel $W(\Lambda_i/\Lambda_{i+1},\sigma^2)$ and $W'(Z;X_i|X_1,\cdot\cdot\cdot,X_{i-1})$ is similar.

We summarize the foregoing analysis in the following lemma:
\begin{lemma}\label{lemma}
Consider a lattice $L$ constructed by a binary lattice partition chain $\Lambda_{1}/\cdot\cdot\cdot/\Lambda_{r}$. Constructing a polar code for the $i$-th equivalent binary-input channel $W'(Z;X_i|X_1,\cdot\cdot\cdot,X_{i-1})$ defined by the right hand side of \eqref{eqn:chain_rule} is equivalent to constructing a polar code for the channel $W(\Lambda_i/\Lambda_{i+1},\sigma^2)$.
\end{lemma}
\begin{remark}
This lemma also holds for the AWGN channel (without the mod-$\Lambda$ front-end). The construction of polar lattices for the AWGN channel are explicitly explained in \cite{yan2}, where nested polar codes are constructed based on a set of $W(\Lambda_i/\Lambda_{i+1},\sigma^2)$ channels. We note that the channel $W(\Lambda_i/\Lambda_{i+1},\sigma^2)$ is degraded with respect to the channel $W(\Lambda_{i+1}/\Lambda_{i+2},\sigma^2)$.
\end{remark}
\begin{remark}
One can expect the equivalence in a more general sense than the construction of polar codes. The proof may be based on the equivalence between coset decoding \cite{BK:Zamir} and maximum likelihood (ML) decoding of the fine lattice $\Lambda_i$ in the presence of the $\Lambda_{i+1}$-aliased Gaussian noise.
\end{remark}

\section{Achieving Strong Secrecy and Reliability}
In this section, we demonstrate the polar lattice constructed from a set of nested polar codes can achieve strong secrecy for the mod-$\Lambda_s$ Gaussian wiretap channel. The construction of component polar codes follows the idea in \cite{polarsecrecy}. We still use the notations in \cite{NewPolarSchemeWiretap} to restate their construction. Define the sets of very reliable and very unreliable indices for a channel $Q$ and $0<\beta<0.5$:
\begin{eqnarray}\label{GoodBadChannel}
\begin{aligned}
\mathcal{G}(Q)&=\{i:Z(Q_{N}^{(i)})\leq 2^{-N^{\beta}}\}, \\
\mathcal{N}(Q)&=\{i:I(Q_{N}^{(i)})\leq 2^{-N^{\beta}}\}.
\end{aligned}
\end{eqnarray}
The following are immediate results of \cite{polarcodes1} and \cite[Lemma 4.7]{polarchannelandsource}:
\begin{eqnarray}
\begin{aligned}
&\lim_{N\rightarrow\infty}|\mathcal{G}(Q)|/N=C(Q), \\ \notag\
&\lim_{N\rightarrow\infty}|\mathcal{N}(Q)|/N=1-C(Q), \notag\
\end{aligned} \notag\
\end{eqnarray}
and since $W$ is degraded with respect to $V$,
\begin{eqnarray}\label{eqn:binarysecrecyrate}
\begin{aligned}
&\lim_{N\rightarrow\infty}|\mathcal{G}(V)\cap\mathcal{N}(W)|/N=C(V)-C(W), \\
&\lim_{N\rightarrow\infty}|\mathcal{G}(V)^{c}\cap\mathcal{N}(W)^{c}|/N=0.
\end{aligned}
\end{eqnarray}
The indices in $\mathcal{G}(V)$ and $\mathcal{N}(W)$ are the reliable and the secure indices. The index set can be partitioned into the following four sets:
\begin{eqnarray}
\begin{aligned}
&\mathcal{A}=\mathcal{G}(V)\cap \mathcal{N}(W) \\ \notag\
&\mathcal{B}=\mathcal{G}(V)\cap \mathcal{N}(W)^{c}\\ \notag\
&\mathcal{C}=\mathcal{G}(V)^{c}\cap \mathcal{N}(W) \\ \notag\
&\mathcal{D}=\mathcal{G}(V)^{c}\cap \mathcal{N}(W)^{c}. \notag\
\end{aligned} \notag\
\end{eqnarray}
Unlike the standard polar coding, the bit-channels are partitioned into three parts: A set $\mathcal{M}$ that carries the confidential message bits, a set $\mathcal{R}$ that carries random bits, and a set $\mathcal{F}$ of frozen bits which are known to both Bob and Eve prior to transmission. It is shown that $\lim_{N\rightarrow\infty}I(M;Z^N)=0$ if we assign the bits as follows:
\begin{eqnarray}\label{eqn:assign}
\begin{aligned}
&\mathcal{A}=\mathcal{M} \\
&\mathcal{B}\subseteq\mathcal{R} \\
&\mathcal{C} \subseteq\mathcal{F} \\
&\mathcal{D} \subseteq\mathcal{R}.
\end{aligned}
\end{eqnarray}
\begin{figure*}[htp]
    \centering
    \includegraphics[width=15cm,height=6cm,keepaspectratio]{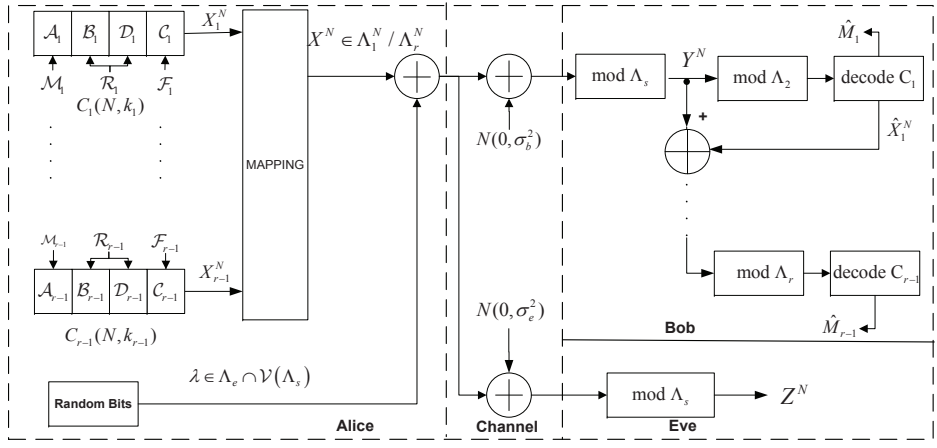}
    \caption{The multilevel lattice coding system over the mod-$\Lambda_s$ Gaussian wiretap channel.}
    \label{fig:eve}
\end{figure*}
\subsection{Gaussian Wiretap Coding Scheme}
Now it is ready to introduce the new polar lattice for the mod-$\Lambda_s$ Gaussian wiretap channel shown in Fig. \ref{fig:eve}. A polar lattice $L$ is constructed by a set of nested polar codes $C_{1}(N,k_{1})\subseteq C_{2}(N,k_{2})\subseteq\cdot\cdot\cdot\subseteq C_{r-1}(N,k_{r-1})$ and a binary lattice partition chain $\Lambda_1/\Lambda_2/\cdot\cdot\cdot/\Lambda_r$. The block length of polar codes is $N$. Alice splits the message $M$ into $M_1,\cdot\cdot\cdot,M_{r-1}$. We follow \eqref{eqn:assign} to assign bits in the component polar codes to achieve strong secrecy. Define $V_i=W(\Lambda_i/\Lambda_{i+1},\sigma_b^2)$ and $W_i=W(\Lambda_i/\Lambda_{i+1},\sigma_e^2)$ and $W_i$ is degraded with respect to $V_i$ for $1\leq i\leq r-1$. Then we can get $\mathcal{A}_i$, $\mathcal{B}_i$, $\mathcal{C}_i$ and $\mathcal{D}_i$ for $1\leq i\leq r-1$. Similarly, we assign the bits as follows
\begin{eqnarray}\label{eqn:newassign}
\begin{aligned}
&\mathcal{A}_i=\mathcal{M}_i \\
&\mathcal{B}_i\subseteq\mathcal{R} \\
&\mathcal{C}_i \subseteq\mathcal{F} \\
&\mathcal{D}_i \subseteq\mathcal{R}
\end{aligned}
\end{eqnarray}
for $1\leq i\leq r-1$. Since $W_i( \text{and }V_i)$ is degraded with respect to $W_{i+1}( \text{and }V_{i+1})$, it is easy to obtain that $\mathcal{C}_i\supseteq\mathcal{C}_{i+1}$ which means $\mathcal{A}_{i}\cup\mathcal{B}_{i}\cup\mathcal{D}_i\subseteq\mathcal{A}_{i+1}\cup\mathcal{B}_{i+1}\cup\mathcal{D}_{i+1}$. This construction is clearly a lattice construction as polar codes constructed on each level are nested. We skip the proof of nested polar codes here. A similar proof can be found in \cite{yan2}.

Interestingly, this polar lattice construction generates an AWGN-good lattice $\Lambda_b$ and a secrecy-good lattice $\Lambda_e$ simultaneously.  $\Lambda_b$ is constructed from a set of nested polar codes $C_{1}(N,|\mathcal{A}_1|+|\mathcal{B}_1|+|\mathcal{D}_1|)\subseteq\cdot\cdot\cdot\subseteq C_{r-1}(N,|\mathcal{A}_{r-1}|+|\mathcal{B}_{r-1}|+|\mathcal{D}_{r-1}|)$ and the lattice partition chain $\Lambda_1/\cdot\cdot\cdot/\Lambda_r$, while $\Lambda_e$ is constructed from a set of nested polar codes $C_{1}(N,|\mathcal{B}_1|+|\mathcal{D}_1|)\subseteq\cdot\cdot\cdot\subseteq C_{r-1}(N,|\mathcal{B}_{r-1}|+|\mathcal{D}_{r-1}|)$ and the same lattice partition chain $\Lambda_1/\cdot\cdot\cdot/\Lambda_r$. More details about the AWGN-goodness of $\Lambda_b$ are given in Sect. IV-C. It is clear that $\Lambda_e\subset\Lambda_b$. Our coding scheme is equivalent to the coset coding scheme we introduced in Sect. I, which maps the confidential message $M$ to the coset leaders $\lambda_m\in\Lambda_b/\Lambda_e$.

\subsection{Strong Secrecy and Secrecy Rate}
By using the above assignments and \cite[Proposition 16]{polarsecrecy}, we have
\begin{eqnarray}
\begin{aligned}
I(M_i;Z_i^N)\leq N2^{-N^\beta},\notag\
\end{aligned} \notag\
\end{eqnarray}
where $Z_i^N=Z^N \text{ mod } \Lambda_{i+1}$. In other words, the employed polar code for the channel $W(\Lambda_i/\Lambda_{i+1},\sigma_e^2)$ can guarantee that the mutual information between the input message and the output is upper bounded by $N2^{-N^\beta}$. From Lemma \ref{lemma}, this polar code can also guarantee the same upper bound on the mutual information between the input message and the output of the channel $W'(Z;X_i|X_1,\cdot\cdot\cdot,X_{i-1})$ as shown in the following inequality:
\begin{eqnarray}
\begin{aligned}
I(M_i;Z^N,M_1,\cdot\cdot\cdot,M_{i-1})\leq N2^{-N^\beta}.\notag\
\end{aligned} \notag\
\end{eqnarray}

Recall $Z^N$ is the signal received by Eve after mod-$\Lambda_r$. From the chain rule of mutual information,
\begin{align}
&I(Z^N;M) \label{eqn:upperbound}\\
&=\sum_{i=1}^{r-1}I(Z^N;M_i|M_1,\cdot\cdot\cdot,M_{r-1}) \notag\\
&=\sum_{i=1}^{r-1}h(M_i|M_1,\cdot\cdot\cdot,M_{r-1})-h(M_i|Z^N,M_1,\cdot\cdot\cdot,M_{r-1})\notag\\
&=\sum_{i=1}^{r-1}h(M_i)-h(M_i|Z^N,M_1,\cdot\cdot\cdot,M_{r-1})\notag\\
&=\sum_{i=1}^{r-1}I(M_i;Z^N,M_1,\cdot\cdot\cdot,M_{i-1})\notag\\
&\leq rN2^{-N^\beta}. \notag\
\end{align}
Therefore strong secrecy is achieved as $\lim_{N\rightarrow\infty}I(M;Z^N)=0$.

Now we present the main theorem of the paper.
\begin{theorem}
Consider a polar lattice $L$ constructed according to \eqref{eqn:newassign} with the binary lattice partition chain $\Lambda_{1}/\cdot\cdot\cdot/\Lambda_{r}$ and $r-1$ binary nested polar codes with block length $N$. By scaling $\Lambda_1$ and $r$ to satisfy the following conditions:
\begin{enumerate}[(i)]
  \item\label{ite:fir} $h(\Lambda_1,\sigma_b^2)\rightarrow \log V(\Lambda_1)$
  \item\label{ite:sec} $h(\Lambda_r,\sigma_e^2)\rightarrow\frac{1}{2}\log(2\pi e\sigma_e^2)$,
\end{enumerate}
given $\sigma_e^2>\sigma_b^2$, as $N\rightarrow\infty$, all strong secrecy rates $R$ satisfying
\begin{eqnarray}
\begin{aligned}
R<\frac{1}{2}\log\frac{\sigma_e^2}{\sigma_b^2}\notag\
\end{aligned} \notag\
\end{eqnarray}
are achievable using the polar lattice $L$ on the mod-$\Lambda_s$ Gaussian wiretap channel.
\end{theorem}
\begin{proof}
By \eqref{eqn:binarysecrecyrate} and \eqref{eqn:newassign},
{\allowdisplaybreaks\begin{eqnarray}\notag\
\begin{aligned}
\lim_{N\rightarrow\infty} R&=\sum_{i=1}^{r-1}\lim_{N\rightarrow\infty}\frac{|\mathcal{A}_i|}{N} \\
&=\sum_{i=1}^{r-1}C(V_i)-C(W_i) \\
&=\sum_{i=1}^{r-1}C(W(\Lambda_i/\Lambda_{i+1},\sigma_b^2))-C(W(\Lambda_i/\Lambda_{i+1},\sigma_e^2)) \\
&=C(W(\Lambda_1/\Lambda_{r},\sigma_b^2))-C(W(\Lambda_1/\Lambda_{r},\sigma_e^2))\\
&=C(\Lambda_r, \sigma_b^2)-C(\Lambda_1, \sigma_b^2)-C(\Lambda_r, \sigma_e^2)+C(\Lambda_1, \sigma_e^2) \\
&=h(\Lambda_r, \sigma_e^2)-h(\Lambda_r, \sigma_b^2)+h(\Lambda_1, \sigma_b^2)-h(\Lambda_1, \sigma_e^2) \\
&=\frac{1}{2}\log\frac{\sigma_e^2}{\sigma_b^2}-(\epsilon_e-\epsilon_b)-\epsilon_1,
\end{aligned}
\end{eqnarray}}
where
\begin{equation} \notag\
\begin{cases} \epsilon_1=h(\Lambda_1, \sigma_e^2)-h(\Lambda_1, \sigma_b^2)\geq0, \\
\epsilon_b=h(\sigma_b^2)-h(\Lambda_r, \sigma_b^2)=\frac{1}{2}\log(2\pi e\sigma_b^2)-h(\Lambda_r, \sigma_b^2)\geq0, \\
\epsilon_e=h(\sigma_e^2)-h(\Lambda_r, \sigma_e^2)=\frac{1}{2}\log(2\pi e\sigma_e^2)-h(\Lambda_r, \sigma_e^2)\geq0
\end{cases}
\end{equation}
and $\epsilon_e-\epsilon_b\geq0$.

$\epsilon_1$ can be made arbitrarily small by scaling $\Lambda_1$ such that both $h(\Lambda_1, \sigma_e^2)$ and $h(\Lambda_1, \sigma_b^2)$ are sufficiently close to $\log V(\Lambda_1)$. In polar lattices for AWGN-goodness \cite{yan2}, we assume $h(\Lambda_{r'}, \sigma_b^2)\approx\frac{1}{2}\log(2\pi e\sigma_b^2)$ for some $r'<r$. Since $\epsilon_b<\epsilon_e$, $\Lambda_{r'}$ is not enough for the wiretap channel. We need to increase the number of levels until $h(\Lambda_r, \sigma_e^2)\approx\frac{1}{2}\log(2\pi e\sigma_e^2)$ such that both $\epsilon_b$ and $\epsilon_e$ are almost 0. Therefore by scaling $\Lambda_1$ and adjusting $r$, the secrecy rate can get arbitrarily close to $\frac{1}{2}\log\frac{\sigma_e^2}{\sigma_b^2}$.
\end{proof}
\begin{remark}
We note that these are just mild conditions. When the $\sigma_e^2=4$ and $\sigma_b^2=1$, the gap from $\frac{1}{2}\log\frac{\sigma_e^2}{\sigma_b^2}$ is only $ 0.05$ when we choose $r=3$ and partition chain $\alpha(\mathbb{Z}/2\mathbb{Z}/4\mathbb{Z})$ with scaling factor $\alpha=2.5$.
\end{remark}
\begin{remark}
$\frac{1}{2}\log\frac{\sigma_e^2}{\sigma_b^2}$ is an upper bound on the secrecy capacity of the mod-$\Lambda_s$ Gaussian wiretap channel since it equals to the secrecy capacity of the Gaussian wiretap channel as the signal power goes to infinity. Latices codes can not be better than this. It is noteworthy that we successfully remove the $\frac{1}{2}$-nat gap in the achievable secrecy rate derived in \cite{cong2} which is caused by the limitation of the $L^{\infty}$ distance associated with the flatness factor.
\end{remark}
\begin{remark}
The two conditions \eqref{ite:fir} and \eqref{ite:sec} are the design criteria of secrecy-good lattices. The construction for secrecy-good lattices requires more levels than the construction of AWGN-good lattices.
\end{remark}

\subsection{Reliability}
Here how to assign $\mathcal{D}$ is a problem. Assigning freezing bits to $\mathcal{D}$ guarantees the reliability but achieves the weak secrecy, whereas assigning random bits to $\mathcal{D}$ guarantees the strong secrecy but may violate the reliability requirement because $\mathcal{D}$ may be nonempty. In the original work on using polar codes to achieve the secrecy capacity for symmetric and binary-input wiretap channels \cite{polarsecrecy}, in order to ensure strong security, $\mathcal{D}$ is assigned with random bits ($\mathcal{D} \in \mathcal{R}$), which results in the fact that this scheme failed to accomplish the theoretical reliability. More explicitly, for any $i$-th level channel $W(\Lambda_i/\Lambda_{i+1}, \sigma_b^2)$ at Bob's end, the probability of error is upper bounded by the sum of the Bhattacharyya parameters $Z(W_N^{(j)}(\Lambda_i/\Lambda_{i+1}, \sigma_b^2))$ of those bit-channels that are not frozen to zero. For each bit-channel index $j$ and $\beta <0.5$, we have
\begin{eqnarray*}
j \in \mathcal{A} \cup \mathcal{R}= \mathcal{G}(W(\Lambda_i/\Lambda_{i+1}, \sigma_b^2), \beta) \cup \mathcal{D}.
\end{eqnarray*}
By the definition \eqref{GoodBadChannel}, we can see that the sum of $Z(W_N^{(j)}(\Lambda_i/\Lambda_{i+1}, \sigma_b^2))$ over the set $\mathcal{G}(W(\Lambda_i/\Lambda_{i+1}, \sigma_b^2)$ is bounded by $2^{-N^{\beta}}$, and therefore, the error probability of the $i$-th level channel under the successive cancellation (SC) decoding, denoted by $P_e^{SC}(\Lambda_i/\Lambda_{i+1}, \sigma_b^2)$, can be upper bounded by
\begin{eqnarray*}
P_e^{SC}(\Lambda_i/\Lambda_{i+1}, \sigma_b^2) \leq 2^{-N^{\beta}} + \sum_{j \in \mathcal{D}} Z(W_N^{(j)}(\Lambda_i/\Lambda_{i+1}, \sigma_b^2)).
\end{eqnarray*}
Since multistage decoding is utilized, by the union bound, the final decoding error probability for Bob is bounded as
\begin{eqnarray*}
\text{Pr}\{\hat{M} \neq M\} \leq \sum_{i=1}^{r-1} P_e^{SC}(\Lambda_i/\Lambda_{i+1}, \sigma_b^2) .
\end{eqnarray*}
Unfortunately, a proof that this scheme satisfies the reliability condition cannot be arrived here because the bound of the sum $\sum_{j \in \mathcal{D}} Z(W_N^{(j)}(\Lambda_i/\Lambda_{i+1}, \sigma_b^2))$ is not known. Note that significantly low probabilities of error can still be achieved in practice since the size of $\mathcal{D}$ is very small.

It is also worthy mentioning that this reliability problem was recently solved in \cite{NewPolarSchemeWiretap}, where a new scheme dividing the information message of each $\Lambda_i/\Lambda_{i+1}$ channel into several blocks is proposed. For a specific block, $\mathcal{D}$ is still assigned with random bits and transmitted in advance in the set $\mathcal{A}$ of the previous block.  This scheme involves negligible rate loss and finally realizes reliability and strong security simultaneously. In this case, if the reliability of each partition channel can be achieved, i.e., for any $i$-th level partition $\Lambda_i/\Lambda_{i+1}$, $P_e^{SC}(\Lambda_i/\Lambda_{i+1}, \sigma_b^2)$ vanishes as $N$ goes to infinity. Then the total decoding error probability for Bob can be made arbitrarily small. Actually, based on the new scheme of assigning the problematic bits in $\mathcal{D}$ \cite{NewPolarSchemeWiretap}, the error probability on level $i$ can be upper bounded by
\begin{eqnarray}
P_e^{SC}(\Lambda_i/\Lambda_{i+1}, \sigma_b^2) \leq \epsilon_{N'}^i +k_i \cdot o(2^{-N'^{\beta}}),
\end{eqnarray}
where $k_i$ is the number of information blocks on the $i$-th level, $N'$ is the length of each block which satisfies $N'\times k_i=N$ and $\epsilon_N^i$ is caused by the first separate block on the $i$-th level consisting of the initial bits in $\mathcal{D}_i$. Since $|\mathcal{D}_i|$ is extremely small comparing to the block length $N$, the decoding failure probability for the first block can be made arbitrarily small when $N$ is sufficiently large. Therefore, $\Lambda_b$ is an AWGN-good lattice.

Note that the rate loss incurred by repeatedly transmitted bits in $\mathcal{D}_i$ is negligible because of its small size and the fact that only one block is wasted on each level. Explicitly, the actually achieved secrecy rate in the $i$-th level is given by $\frac{k_i}{k_i+1} [C(\Lambda_i/\Lambda_{i+1}, \sigma_b^2)-C(\Lambda_i/\Lambda_{i+1}, \sigma_e^2)]$. Clearly, this rate can be made close to the maximum secrecy rate by choosing sufficiently large $k_i$ as well.

\section{Concluding Remarks}
In this work, we showed that the new polar lattice can achieve both strong secrecy and reliability over the mod-$\Lambda_s$ Gaussian wiretap channel. The uniform assumption was used to construct AWGN-good lattice $\Lambda_b$ and the secrecy-good lattice $\Lambda_e$. Since the confidential message is mapped to the coset leaders of $\Lambda_b/\Lambda_e$, the channel between the confidential message and the Eavesdropper can be seen as a $\Lambda_b/\Lambda_e$ channel. Since the $\Lambda_b/\Lambda_e$ channel is symmetric, the maximum mutual information is achieved by the uniform input. Consequently, the mutual information corresponding to other input distributions can also be upper bounded by $rN2^{-N^\beta}$ in \eqref{eqn:upperbound}. Therefore, strong secrecy under other distributions can also be proved, which means our scheme actually achieves semantical security in cryptographic terms \cite{cong2,semantic1}. The rigorous proof will be given in the journal paper. A similar statement for the binary symmetric channel can be found in \cite[Theorem 4.12]{semantic1}.

\section*{Acknowledgments}
The authors would like to thank Prof. Jean-Claude Belfiore for helpful discussions. This work was supported in part by FP7 project PHYLAWS (EU FP7-ICT 317562) and in part by the China Scholarship Council.

\bibliographystyle{IEEEtran}
\bibliography{yanfei}

\begin{thebibliography}{10}
\providecommand{\url}[1]{#1}
\csname url@samestyle\endcsname
\providecommand{\newblock}{\relax}
\providecommand{\bibinfo}[2]{#2}
\providecommand{\BIBentrySTDinterwordspacing}{\spaceskip=0pt\relax}
\providecommand{\BIBentryALTinterwordstretchfactor}{4}
\providecommand{\BIBentryALTinterwordspacing}{\spaceskip=\fontdimen2\font plus
\BIBentryALTinterwordstretchfactor\fontdimen3\font minus
  \fontdimen4\font\relax}
\providecommand{\BIBforeignlanguage}[2]{{%
\expandafter\ifx\csname l@#1\endcsname\relax
\typeout{** WARNING: IEEEtran.bst: No hyphenation pattern has been}%
\typeout{** loaded for the language `#1'. Using the pattern for}%
\typeout{** the default language instead.}%
\else
\language=\csname l@#1\endcsname
\fi
#2}}
\providecommand{\BIBdecl}{\relax}
\BIBdecl

\bibitem{wyner1}
A.~D. Wyner, ``The wire-tap channel,'' \emph{Bell System Technical Journal},
  vol.~54, no.~8, pp. 1355--1387, October 1975.

\bibitem{csis1}
I.~Csisz\'{a}r, ``Almost independence and secrecy capacity,'' \emph{Problems of
  Information Transmission}, vol.~32, pp. 40--47, 1996.

\bibitem{polarcodes}
E.~Ar{\i}kan, ``Channel polarization: A method for constructing
  capacity-achieving codes for symmetric binary-input memoryless channels,''
  \emph{IEEE Trans. Inform. Theory}, vol.~55, no.~7, pp. 3051--3073, July 2009.

\bibitem{polarsecrecy}
H.~Mahdavifar and A.~Vardy, ``Achieving the secrecy capacity of wiretap
  channels using polar codes,'' \emph{IEEE Trans. Inform. Theory}, vol.~57,
  no.~10, pp. 6428--6443, Oct. 2011.

\bibitem{NewPolarSchemeWiretap}
E.~\c{S}a\c{s}o\v{g}lu and A.~Vardy, ``A new polar coding scheme for strong
  security on wiretap channels,'' in \emph{Proc. 2013 IEEE Int. Symp. Inform.
  Theory (ISIT 2013)}, 2013, pp. 1117--1121.

\bibitem{cong}
L.-C. Choo, C.~Ling, and K.-K. Wong, ``Achievable rates for lattice coding over
  the {G}aussian wiretap channel,'' \emph{ICC 2011 Physical Layer Security
  Workshop}, 2011.

\bibitem{cong2}
\BIBentryALTinterwordspacing
C.~Ling, L.~Luzzi, J.-C. Belfiore, and D.~Stehl\'{e}, ``Semantically secure
  lattice codes for the {G}aussian wiretap channel,'' 2012. [Online].
  Available: \url{http://arxiv.org/abs/1210.6673}
\BIBentrySTDinterwordspacing

\bibitem{belf3}
\BIBentryALTinterwordspacing
F.~E. Oggier, P.~Sol{\'e}, and J.-C. Belfiore, ``Lattice codes for the wiretap
  {G}aussian channel: Construction and analysis,'' vol. abs/1103.4086, Mar.
  2011. [Online]. Available: \url{arXiv:1103.4086v1[cs.IT]}
\BIBentrySTDinterwordspacing

\bibitem{BK:Zamir}
R.~Zamir, \emph{Lattice Coding for Signals and Networks}.\hskip 1em plus 0.5em
  minus 0.4em\relax Cambridge, UK: Cambridge University Press, book in
  preparation.

\bibitem{yan2}
Y.~Yan, C.~Ling, and X.~Wu, ``Polar lattices: {W}here {A}r{\i}kan meets
  {F}orney,'' in \emph{Proc. 2013 IEEE Int. Symp. Inform. Theory (ISIT 2013)},
  2013, pp. 1292--1296.

\bibitem{forney6}
G.~D. Forney~Jr., M.~Trott, and S.-Y. Chung, ``Sphere-bound-achieving coset
  codes and multilevel coset codes,'' \emph{IEEE Trans. Inform. Theory},
  vol.~46, no.~3, pp. 820--850, May 2000.

\bibitem{poltyrev}
G.~Poltyrev, ``On coding without restictions for the {AWGN} channel,''
  \emph{IEEE Trans. Inform. Theory}, vol.~40, pp. 409--417, Mar. 1994.

\bibitem{yellowbook}
J.~H. Conway and N.~J.~A. Sloane, \emph{Sphere packings, lattices, and groups},
  Second Edition, 1993, Springer-Verlag, New York.

\bibitem{multilevel1}
R.~Fischer, ``The modulo-lattice channel: The key feature in precoding
  schemes,'' \emph{International Journal of Electronics and Communications
  (AE\"{U})}, vol.~59, no.~4, pp. 244--253, June 2005.

\bibitem{multilevel}
U.~Wachsmann, R.~Fischer, and J.~Huber, ``Multilevel codes: Theoretical
  concepts and practical design rules,'' \emph{IEEE Trans. Inform. Theory},
  vol.~45, no.~5, pp. 1361--1391, July 1999.

\bibitem{polarcodes1}
E.~Ar{\i}kan and E.~Telatar, ``On the rate of channel polarization,'' in
  \emph{Proc. 2009 IEEE Int. Symp. Inform. Theory (ISIT 2009)}, 2009, pp.
  1493--1495.

\bibitem{polarchannelandsource}
S.~B. Korada, ``Polar codes for channel and source coding,'' Ph.D.
  dissertation, Ecole Polytechnique F\'{e}d\'{e}rale de Lausanne, 2009.

\bibitem{semantic1}
M.~Bellare, S.~Tessaro, and A.~Vardy, ``A cryptographic treatment of the
  wiretap channel,'' \emph{Cryptology ePrint Archive, Report 2012/015}, Jan.
  2012.

\end{thebibliography}

\end{document}